\documentclass[letterpaper, 10 pt, conference]{ieeeconf}  

\IEEEoverridecommandlockouts                              
\usepackage{cite}
\usepackage{amsmath,amssymb,amsfonts}
\usepackage{algorithmic}
\usepackage{graphicx}
\usepackage{textcomp}
\usepackage{xcolor}
\usepackage[scr]{rsfso}
\usepackage{authblk}
\usepackage{float}
\usepackage{soul}
\usepackage{hyperref}
\usepackage{url}
\usepackage{float}
\usepackage{titlesec}
\usepackage{booktabs} 
\usepackage{array}    
\usepackage{multirow} 
\usepackage{siunitx}  
\usepackage{caption}  
\usepackage{makecell} 
\usepackage{mathtools}
\usepackage{etoolbox} 
\usepackage{graphicx}
\usepackage{svg}
\usepackage{subcaption}
\usepackage{bm}

\usepackage{hyperref}
\hypersetup{
    colorlinks=true,
    linkcolor=blue,
    citecolor=red,
    urlcolor=magenta
}
\usepackage{lmodern}
\usepackage{titlesec}

\newtheorem{theorem}{Theorem}
\newtheorem{definition}[theorem]{Definition}
\newtheorem{assumption}[theorem]{Assumption}
\newtheorem{lemma}[theorem]{Lemma}

\usepackage{xcolor}

\newcounter{parcount}[section]
\renewcommand{\theparcount}{\textit{(\alph{parcount})}}

\newcommand{\myparagraph}[1]{%
  \par\refstepcounter{parcount}%
  \noindent\hspace*{1em}\theparcount \textit{ #1}:
}

\makeatletter
\@addtoreset{parcount}{section}
\makeatother

\begin{document}
\title{
Delay compensation of multi-input distinct delay \\ nonlinear systems via
neural operators}

\author[1]{Filip Bajraktari}
\author[2]{Luke Bhan\thanks{Correspondence to: lbhan@ucsd.edu}}
\author[2]{Miroslav Krstic}
\author[2]{Yuanyuan Shi}

\affil[1]{University of Belgrade, Serbia}
\affil[2]{University of California San Diego, La Jolla, CA, USA}

\makeatletter
\patchcmd{\@maketitle}{\AB@authlist}{\AB@authlist\vspace{-1.5em}}{}{}
\makeatother

\maketitle
\begin{abstract}
In this work, we present the first stability results for approximate predictors in multi-input nonlinear systems with distinct actuation delays. We show that if the predictor approximation satisfies a uniform (in time) error bound, semi-global practical stability is correspondingly achieved.  For such approximators, the required uniform error bound depends on the desired region of attraction and the number of control inputs in the system. The result is achieved through transforming the delay into a transport PDE and conducting analysis on the coupled ODE-PDE cascade. To highlight the viability of such error bounds, we demonstrate our results on a class of approximators - neural operators - showcasing sufficiency for satisfying such a universal bound both theoretically and in simulation on a mobile robot experiment. 
\vspace*{\baselineskip}

\end{abstract}

\section{Introduction} \label{sec:introduction}
We study the multi-input, distinct delay, nonlinear system governed by
\begin{equation}
    \dot{X}(t) = f \bigg( X(t), U_1(t-D_1), \dots, U_m(t-D_m) \bigg), \label{eq:system}
\end{equation}
where $X \in \mathbb{R}^n$ is the system state, $U_1, \dots, U_m$ are scalar control inputs, $f: \mathbb{R}^n \times \mathbb{R}^m \rightarrow \mathbb{R}^n$ is a locally Lipschitz vector field satisfying $f(0, \dots,0) = 0$, and $D_1, \dots, D_m$ are (potentially distinct) input delays satisfying (without the loss of generality) $0 < D_1 \leq \dots \leq D_m$. This system commonly appears across a range of applications, including telerobotics \cite{9727201}, connected and automated vehicles \cite{samii2025experimentalimplementationvalidationpredictorbased}, networking systems \cite{doi:10.1137/140980570}, and many more \cite{Bekiaris-Liberis_Krstic:2016}.

To compensate for the delays, the most popular approach for nonlinear systems, termed ``predictor feedback", was developed in \cite{Krstic:2010}. The key idea of predictor feedback designs is to apply known delay-free, stabilizing controllers with a future prediction of the state; thus, when the control law arrives at the system, it aligns with the system's true state. Consequently, for any predictor feedback design, we  require the existence of a set of globally asymptotically stabilizing delay-free control laws:
\begin{assumption}\label{assumption:delay-free-control-laws}
There exists a set of controllers $U_i(t) = \kappa_i \big( X(t)\big)$ with $\kappa_i \in C^1(\mathbb{R}^n;\mathbb{R})$, $i \in [m]$, such that the delay-free system $\dot{X}(t) = f \bigg( X(t), U_1(t), \dots, U_m(t) \bigg)$ is globally asymptotically stable.
\end{assumption}
\noindent In multi-input multi-delay systems, a key challenge naturally arises: with multiple inputs and associated delays, it becomes difficult to coordinate the predictors so that all $m$ control inputs arrive with the correct corresponding state predictions. This was resolved for \emph{exact} predictor feedback designs in \cite{Bekiaris-Liberis_Krstic:2016} with predictors defined as
\allowdisplaybreaks
\begin{subequations}\label{eqs:P_system}
    \begin{align}
        P_1(t) =& X(t) + \int_{t-D_1}^t f \bigg(P_1(\theta), U_1(\theta), U_2(\theta - D_{2,1}),\nonumber \\ 
        &   \dots, U_m(\theta-D_{m,1}) \bigg) d\theta \label{eq:P1(t)} \\
        P_2(t) =& P_1(t) + \int_{t-D_{2,1}}^t f \bigg( P_2(\theta), \kappa_1(P_2(\theta)), U_2(\theta), \nonumber \\& U_3(\theta - D_{3,2}),\dots, U_m(\theta-D_{m,2}) \bigg) d\theta \label{eq:P2(t)} \\
        &\hspace{3cm} \vdots \nonumber \\
        P_m(t) =& P_{m-1}(t) + \int_{t-D_{m,m-1}}^t f \bigg( P_m(\theta), \kappa_1(P_m(\theta)), \nonumber \\ & \kappa_2(P_m(\theta)), \dots, \kappa_{m-1}(P_m(\theta)), U_m(\theta) \bigg) d\theta, \label{eq:Pm(t)}
    \end{align}
\end{subequations}
\allowdisplaybreaks
with initial conditions given by 
\begin{subequations}\label{eqs:P_init_system}
\begin{alignat}{2}
  P_1(\theta) 
   = &X(0) + \int_{-D_1}^\theta f\big( P_1(s), U_1(s), U_s(&&\!\!\!\!s- D_{2,1}), \nonumber  \\
  &\dots, U_m(s - D_{m,1}) \big) ds&&  \\ 
  & && \hspace{-0.5cm} -D_1 \leq \theta \leq 0,\nonumber  \\
P_2(\theta) 
    =& P_1(0) + \int_{-D_{2,1}}^\theta f\big( P_2(s), \nonumber \kappa_1(P_2(s)&&\!\!\!\!), U_2(s), 
 \\&U_3(s - D_{3,2}), \dots, U_m(s - D_{m,2}) &&\!\!\!\!\big) ds \\ 
  & && \hspace{-0.7cm} -D_{2,1} \leq \theta \leq 0, \nonumber \\
  & \hspace{3.5cm} \vdots \nonumber  && \\ 
P_m(\theta) =& P_{m-1}(0) + \int_{-D_{m,m-1}}^\theta f \bigg( P_m(s), \kappa_1&&(P_m(s)), \nonumber \\ &\kappa_2(P_m(s)), \dots, \kappa_{m-1}(P_m(s)), U_m&&(s) \bigg) ds \label{eq:Pm(theta)}\,, \\
            & && \hspace{-1.3cm}  -D_{m,m-1} \leq \theta \leq 0\,, \nonumber
    \end{alignat}
\end{subequations}
where $D_{i, j}$ represents the difference between delays, $D_i - D_j$. 

While the predictors in \eqref{eqs:P_system}, \eqref{eqs:P_init_system} address the issue of distinct delays, the original global stability result \cite[Theorem 1]{Bekiaris-Liberis_Krstic:2016} only applies to exact implementations—which are never feasible without closed-form solutions. Thus, as implicit ODEs, implementing the predictors numerically requires both fixed-point iteration and integral discretization \cite{Predictor_Feedback_For_Delay_Systems}. This poses two main challenges: (i) discretization errors grow with step size and delay, potentially destabilizing the feedback, and (ii) convergence of the fixed-point iteration is not guaranteed, especially for stiff systems. To address this, we make two contributions: (1) a stability result for approximate predictors with uniform error bounds, yielding a semi-global region of attraction; and (2) a neural network-based predictor approximation that significantly accelerates computation compared to traditional methods. For the context, we begin with an overview of recent works in predictor feedback and learning-based approximator designs.

\myparagraph{Delay compensation in nonlinear systems}
Since the Smith predictor \cite{smith1957closer}, predictor feedback has been a foundational tool for delay compensation in nonlinear systems. It has been successfully applied in diverse settings, including delay-adaptive control \cite{6704718}, PDE-based control \cite{9528941}, and event-triggered control \cite{NOZARI2020108754}. In this work, we focus on the predictor feedback designs introduced in \cite{Krstic:2010, Nonlinear_Control_Under_Nonconstant_Delay}, which model input delays as transport PDEs and apply backstepping transformations to guarantee stability. This PDE-based approach has enabled advances in controlling systems with delays across several domains, including connected and automated vehicles (CAVs) \cite{samii2025experimentalimplementationvalidationpredictorbased}, biological systems \cite{6704718}, and additive manufacturing \cite{7921611}. For a comprehensive overview of predictor feedback methods, see \cite{Deng10092022}.

\myparagraph{Learning-based approximations in control}
For decades, neural networks hve been approximating various computationally intractable portions of control. Perhaps most famously, these types of approaches have been long studied for the Hamilton-Jacobi-Bellman equations \cite{4359214}, as well as ODE approximations \cite{SONTAG1998151}, differential games \cite{pmlr-v283-sharpless25a}, and even in the context of operator approximations in PDEs \cite{Luke_gain:2024, KRSTIC2024111649}. They were recently introduced for predictors in \cite{Luke_Peijia:2025} and have since been extended for unknown delays in \cite{CDC}. We emphasize that there are many many more results not present here, but due to space constraints, must be omitted.

\myparagraph{Notation}
We use $\mathbb{R}_+$ to denote the set of positive reals. For an $n$-vector, the norm $|\cdot|$ denotes the usual Euclidean norm. For a function $u: [0,1]\times\mathbb{R} \rightarrow \mathbb{R}$, we denote by $\|u(t)\|_{L^\infty}$ the supremum spatial norm ($\|u(t)\|_{L^\infty} = \sup_{x \in [0,1]}|u(x,t)|$). We use the PDE notation $u_x(x,t)=\frac{\partial u}{\partial x}(x,t)$ to denote derivatives. We use $C^n([t-D,t];\mathbb{R}^m)$ to denote the set of functions with continuous $n$ derivatives. For brevity, when applicable we use $i \in [m]$ to represent the iteration $i \in \{1, 2, 3, \cdots , m\}$.

\section{An operator theoretic view of predictors} \label{sec:predictor-operator}
In this section, we begin by establishing a formal mapping of the solution to the implicit predictors in \eqref{eqs:P_system}, \eqref{eqs:P_init_system}. This perspective, first introduced in \cite{Luke_Peijia:2025}, enables us to view the solution to the implicit predictor ODE as a mapping of the current state and $m$ different control histories ($U_i$, $i \in [m]$) into an estimate of the system at $X(t+D_i)$. However, this definition was originally introduced for a single-input scalar system and thus requires extension to the multi-input multi-delay case. Henceforth, we begin by introducing the definition of $m$ different predictor operators
\begin{definition}\label{def:delay_dependent_predictor_operators}
(\textbf{Predictor operators for multi-input multi-delay systems}) Let $Q \in \mathbb{R}^n$, and for each $i \in [m]$, let $U_i \in C^1([0,1];\mathbb{R})$. Let $\varphi \in \mathbb{R}_+$. Then, define the set of predictor operators $\{\mathcal{P}_1, ..., \mathcal{P}_m\}$ mapping from $\mathbb{R}^n \times C^1([0,1];\mathbb{R})^{m-i+1} \times \mathbb{R}_+$ to $C^1([0,1];\mathbb{R}^n)$ taking
 inputs $(Q,U_i,\dots,U_m, \varphi)$ and returning a function $P_i \in C^1([0,1];\mathbb{R}^n)$ where the output $P_i$ satisfies the implicit differential equation
\begin{align}
    0 = P_i(s) - Q& - \varphi \int_0^s f(P_i(\theta), \kappa_1(P_i(\theta)), \dots, \kappa_{i-1}(P_i(\theta)),\nonumber \\ 
    &U_i(\theta), \dots, U_m(\theta)) d\theta , \quad s \in [0,1], \label{eq:predictor_def_Pi}
\end{align}
where $\kappa_i$, $i \in [m-1]$, are a set of $C^1$ functions from $\mathbb{R}^n \to \mathbb{R}$ that are assumed to be apriori known as static parameters as in Assumption \ref{assumption:delay-free-control-laws}.
\end{definition}
Notice that, by definition, the predictor operator yields the solution to \eqref{eqs:P_system} in the following sense
\begin{subequations}
\begin{align}
    P_1(t) =& \mathcal{P}_1(X(t), T_{D_1,0}(t)U_1, ..., \nonumber \\ & \quad  T_{D_m,D_{m-1}}(t)U_m, D_1)(1), \\
    P_2(t) =& \mathcal{P}_2(P_1(t), T_{D_{2,1},0}(t)U_2, ..., \nonumber \\ & \quad T_{D_{m,1},D_{m,2}}(t)U_m, D_{2,1})(1), \\
    & \hspace{2.5cm} \vdots \nonumber \\
    P_m(t) =& \mathcal{P}_m(P_{m-1}(t), T_{D_{m,m-1},0}(t)U_m, D_{m,m-1})(1),
\end{align}
\end{subequations}
where $T_{\varphi_i,\varphi_j}$, $i, j \in [m]$, $i \geq j$, are the operators that yield the control history of $U$ defined as
\begin{align}
    (T_{\varphi_i,\varphi_j}(t)U) := U(t-\varphi_i+(\varphi_i-\varphi_j)x), \quad x \in [0,1).
\end{align}

The choice to use $m$ separate predictor operators instead of one large operator offers is both deliberate and advantageous. It allows training $m$ separate neural networks on specific datasets rather than a single monolithic model, simplifying training and enabling parallel inference across compute nodes for significant speedups. The downside, however, is that this modular design, causes output discontinuities at operator boundaries as 
$P_{j}(\cdot)(1) \neq P_{j+1}(\cdot)(0)$, 
$j \in [m-1]$. 
Consequently, this is only a theoretical inhibition as practicality, one can enforce continuity in implementation, setting $P_{j+1}(\cdot)(0) = P_{j}(\cdot)(1)$ or following the nested operator design \cite{D2EE04204E}. Moreover, by design, only the minimal information for each predictor—controlled by the deliberate shift $T_{\varphi_1, \varphi_2}(t)$—is needed, reducing computational load.  Lastly, scaling predictors on $s \in [0, 1]$ is intentional, allowing reuse without retraining when the input context $\varphi$ changes and thus, enables scaling across different delays.

\subsection{Continuity of the predictor operators} \label{subsec:continuity}
 In anticipation of approximating the predictor operators in Section \ref{subsec:neural-op-approximators}, we showcase that the operator is Lipschitz continuous.
 To do so, it is standard in the predictor feedback literature \cite{Predictor_Feedback_For_Delay_Systems, Luke_Peijia:2025} to assume Lipschitz continuity of the dynamics over an arbitrary large, but compact region of the global state space.

\begin{assumption}\label{assumption:lipschitz}
Let $f(X,U_1,\dots,U_m)$ be as in \eqref{eq:system} and $\mathcal{X} \subset \mathbb{R}^n$, $\mathcal{U} \subset \mathbb{R}$ be two compact sets. Then, assume there exists a constant $C_f > 0$ such that $f$ satisfies the Lipschitz condition
\begin{align}
    \big| f(x_1,u_1,\dots,u_m&)-f(x_2,u_{m+1},\dots,u_{2m}) \big|\nonumber  &&\\ \hspace{1cm} \leq& C_f \bigg( |x_1-x_2| + \sum_{i=1}^m |u_i - u_{m+i}| \bigg),
\end{align}
for all $x_1,x_2 \in \mathcal{X}$ and $u_1,\dots,u_{2m} \in \mathcal{U}$.
\end{assumption}

Now, we prove the Lipschitz continuity of the predictor operator in the following lemma.

\begin{lemma}\label{lemma:predictor-continuity}
Let $\mathcal{X}$ and $\mathcal{U}$ be compact sets and $C_f$ be the Lipschitz constant as in Assumption~\ref{assumption:lipschitz}. Let $X_1, X_2 \in \mathcal{X}$, control functions $U_j, U_{m+j} \in C^1([0,1];\mathcal{U})$ with $j \in [m]$, and positive real numbers $\varphi_1, \varphi_2 \in \Phi \subset \mathbb{R}_+$ where $\Phi$ is compact. Further, let $\overline{\mathcal{X}}, \overline{\mathcal{U}}$, and $\overline{\Phi}$ be upper bounds for each of their underlying sets. Then, each predictor operator $\mathcal{P}_i$, $i \in [m]$, from Definition~\ref{def:delay_dependent_predictor_operators} satisfies the following Lipschitz condition
\begin{align}
    \bigg\| \mathcal{P}_i&(X_1,U_i,\dots,U_m,\varphi_1) - \mathcal{P}_i(X_2,U_{m+i},\dots,U_{2m},\varphi_2) \bigg\|_{L^\infty} \nonumber  \\
    \leq& C_{\mathcal{P}} \Biggl( |X_1 - X_2| + \sum_{j=i}^m \bigg\| U_j - U_{m+j} \bigg\|_{L^\infty} 
    + |\varphi_1 - \varphi_2| \Biggr), \label{eq:predictor_lipschitz_inequality}
\end{align}
with Lipschitz constant
\begin{align}
    C_{\mathcal{P}} :=& \max \bigg( 1, \Xi, \overline{\Phi} C_f \bigg) e^{\overline{\Phi} C_{\kappa}}, \\ 
    \Xi  :=& m C_f \overline{\mathcal{U}} + C_{\kappa} \bigg( \overline{\mathcal{X}} + m \overline{\Phi} C_f \overline{\mathcal{U}} \bigg) e^{\overline{\Phi} C_{\kappa}}, \\
    C_{\kappa}  :=& C_f \bigg( 1 + \sum_{i=1}^m C_{\kappa_i} \bigg), \label{eq:c_kappa}
\end{align}
where $C_{\kappa_i}$ is Lipschitz constant of each control law $\kappa_i$.
\end{lemma}
\begin{proof}
Define the shorthand notation
\begin{align}
    P_i^{(1)}(s) &:= \mathcal{P}_i(X_1,U_i,\dots,U_m, \varphi_1)(s) \\
    P_i^{(2)}(s) &:= \mathcal{P}_i(X_2,U_{m+i},\dots,U_{2m}, \varphi_2)(s),
\end{align}
for all $s \in [0,1]$. Before we proceed to prove the main part of the lemma, we show that each predictor is uniformly bounded in its domain. First, by definition, Assumption's \ref{assumption:delay-free-control-laws}, \ref{assumption:lipschitz} and the compactness of $\mathcal{X}$, we have
\begin{align}
    \bigg| P_i^{(1)}(s) \bigg| \leq& |X_1| + \varphi_1 \int_{0}^s \bigg| f \bigg( P_i^{(1)}(\theta), \kappa_1 \bigg( P_i^{(1)}(\theta) \bigg)\nonumber \\ 
    &\quad , \dots, \kappa_{i-1} \bigg( P_i^{(1)}(\theta) \bigg), U_i(\theta), \dots, U_m(\theta) \bigg) \bigg| d\theta \nonumber \\
        \leq& |X_1| + \varphi_1 \int_{0}^s C_f \Biggl( \bigg| P_i^{(1)}(\theta) \bigg| + \sum_{j=1}^{i-1} C_j \bigg| P_i^{(1)}(\theta) \bigg| \nonumber \\ & + \sum_{j=i}^m \bigg| U_j(\theta) \bigg| \Biggr) d\theta \nonumber \\
        \leq& |X_1| + \varphi_1 C_f \sum_{j=i}^m \| U_j \|_{L^\infty} \nonumber \\ & + \varphi_1 C_{\kappa} \int_{0}^s \bigg| P_i^{(1)}(\theta) \bigg| d\theta.
\end{align}
Applying Gronwall's inequality yields
\begin{align}
    \bigg| P_i^{(1)}(s) \bigg| \leq& \bigg( \overline{\mathcal{X}} + m \overline{\Phi} C_f \overline{\mathcal{U}} \bigg) e^{\overline{\Phi} C_{\kappa}}.
\end{align}
Then, we have
\begin{align}
    \bigg| P_i^{(1)}&(s) - P_i^{(2)}(s) \bigg| \nonumber \\ 
        \leq& |X_1 - X_2|  + \int_{0}^s \bigg| \varphi_1 f \bigg[ P_i^{(1)}(\theta), \kappa_1 \bigg( P_i^{(1)}(\theta) \bigg),\nonumber \\ 
            & \quad \dots, \kappa_{i-1} \bigg( P_i^{(1)}(\theta) \bigg), U_i(\theta), \dots, U_m(\theta) \bigg] \nonumber \\
            & \quad - \varphi_2 f \bigg[ P_i^{(2)}(\theta), \kappa_1 \bigg( P_i^{(2)}(\theta) \bigg), \dots, \kappa_{i-1} \bigg( P_i^{(2)}(\theta) \bigg), \nonumber \\ 
            & \quad U_{m+i}(\theta), \dots, U_{2m}(\theta) \bigg] \bigg| d\theta \nonumber \\
        \leq& |X_1 - X_2|+ |\varphi_1 - \varphi_2| \int_{0}^s \bigg| f \bigg( P_i^{(1)}(\theta), \kappa_1 \bigg( P_i^{(1)}(\theta) \bigg),\nonumber \\ 
            & \quad \dots, \kappa_{i-1} \bigg( P_i^{(1)}(\theta) \bigg), U_i(\theta), \dots, U_m(\theta) \bigg) \bigg| d\theta \nonumber \\
            &\quad + \varphi_2 \int_{0}^s \Bigl| f \bigg( P_i^{(1)}(\theta), \kappa_1 \bigg( P_i^{(1)}(\theta) \bigg), \dots,\nonumber \\ 
            & \quad \kappa_{i-1} \bigg( P_i^{(1)}(\theta) \bigg), U_i(\theta), \dots, U_m(\theta) \bigg) \nonumber \\
            &\qquad - f \bigg( P_i^{(2)}(\theta), \kappa_1 \bigg( P_i^{(2)}(\theta) \bigg), \dots, \kappa_{i-1} \bigg( P_i^{(2)}(\theta) \bigg), \nonumber \\ 
            & \quad U_{m+i}(\theta), \dots, U_{2m}(\theta) \bigg) \Bigr| d\theta \nonumber \\
        \leq& |X_1 - X_2| \nonumber \\
            &\quad + |\varphi_1 - \varphi_2| \bigg[ m C_f \overline{\mathcal{U}} + C_{\kappa} \bigg( \overline{\mathcal{X}} + m \overline{\Phi} C_f \overline{\mathcal{U}} \bigg) e^{\overline{\Phi} C_{\kappa}} \bigg] \nonumber \\
            &\quad + \varphi_2 \int_{0}^s C_f \Biggl( \bigg| P_i^{(1)}(\theta) - P_i^{(2)}(\theta) \bigg| \nonumber \\
            &\quad  + \sum_{j=1}^{i-1} C_j \bigg( P_i^{(1)}(\theta) - P_i^{(2)}(\theta) \bigg) \nonumber \\ 
            & \quad + \sum_{j=i}^m \bigg| U_j(\theta) - U_{m+j}(\theta) \bigg| \Biggr) d\theta \nonumber \\
        \leq& |X_1 - X_2| \nonumber \\
            &\quad + |\varphi_1 - \varphi_2| \bigg[ m C_f \overline{\mathcal{U}} + C_{\kappa} \bigg( \overline{\mathcal{X}} + m \overline{\Phi} C_f \overline{\mathcal{U}} \bigg) e^{\overline{\Phi} C_{\kappa}} \bigg] \nonumber \\
            &\quad + \varphi_2 C_f \sum_{j=1}^m \| U_j - U_{m+j} \|_{L^\infty} \nonumber \\
            &\qquad + \varphi_2 C_{\kappa} \int_0^s \bigg| P_i^{(1)}(\theta) - P_i^{(2)}(\theta) \bigg| d\theta, \label{eq:setup_for_gronwall}
\end{align}
Applying Gronwall's inequality again yields the desired result.
\end{proof}

Notice that Lemma \ref{lemma:predictor-continuity} explicitly depends on the underlying size of the spaces of the input functions which is as expected in operator approximation continuity bounds (cf. \cite[Lemma 1]{Luke_gain:2024}, \cite[Lemma 2]{Max_gain_schedule:2025}). We are now ready to present our class of approximate predictors via neural operators. 



\subsection{Neural operator based approximate predictors} \label{subsec:neural-op-approximators}
Neural operators are neural networks defined by a specific kernel structure such that they have universal operator approximation guarantees (See \cite{Neural_Operator:2023}). There are a variety of different architectures including Deep Operator Networks (DeepONet) \cite{deeponet:2021}, Fourier Neural Operator (FNO) \cite{FNO:2021}, and even single-head attention transformers \cite{Neural_Operator:2023}. Formally, they are defined as 
\begin{definition}[Neural Operators] \label{definition:neural-operator} \cite[Section 1.2]{lanthaler2024nonlocalitynonlinearityimpliesuniversality} Let $\Omega_u \subset \mathbb{R}^{d_{u_1}}$, $\Omega_v \subset \mathbb{R}^{d_{v_1}}$ be bounded domains with Lipschitz boundary and let  $\mathcal{F}_c \subset C^0(\Omega_u; \mathbb{R}^c)$, $\mathcal{F}_v \subset C^0(\Omega_v; \mathbb{R}^v)$ be continuous function spaces.
    Given a channel dimension $d_c > 0$, we call any $\hat{\Psi}$ a neural operator given it satisfies the compositional form $\hat{\Psi} = \mathcal{Q} \circ \mathcal{L}_L \circ \cdots \circ \mathcal{L}_1 \circ \mathcal{R}$ where  $\mathcal{R}$ is a lifting layer, $\mathcal{L}_l, l=1,..., L$ are the hidden layers, and $\mathcal{Q}$ is a projection layer. That is, 
    $\mathcal{R}$ is given by 
    \begin{equation}
    \mathcal{R} : \mathcal{F}_c(\Omega_u; \mathbb{R}^c) \rightarrow \mathcal{F}_s(\Omega_s; \mathbb{R}^{d_c}), \quad c(x) \mapsto R(c(x), x)\,, 
\end{equation} where $\Omega_s \subset \mathbb{R}^{d_{s_1}}$, $\mathcal{F}_s(\Omega_s; \mathbb{R}^{d_c})$ is a Banach space for the hidden layers and $R: \mathbb{R}^c \times \Omega_u \rightarrow \mathbb{R}^{d_c}$ is a learnable neural network acting between finite-dimensional Euclidean spaces. For $l=1, ..., L$, each hidden layer is given by 
\begin{equation} \label{eq:generalNeuralOperator}
    (\mathcal{L}_l v)(x) := s \left( W_l v(x) + b_l + (\mathcal{K}_lv)(x)\right)\,, 
\end{equation}
where weights $W_l \in \mathbb{R}^{d_c \times d_c}$ and biases $b_l \in \mathbb{R}^{d_c}$ are learnable parameters, $s: \mathbb{R} \rightarrow \mathbb{R}$ is a smooth, infinitely differentiable activation function that acts component wise on inputs and $\mathcal{K}_l$ is the nonlocal operator given by 
\begin{equation} \label{eq:generalKernel}
    (\mathcal{K}_lv)(x) = \int_\mathcal{X} K_l(x, y) v(y) dy\,,
\end{equation}
where $K_l(x, y)$ is a kernel function containing learnable parameters. Lastly, the projection layer $\mathcal{Q}$ is given by 
\begin{equation}
    \mathcal{Q} : \mathcal{F}_s(\Omega_s; \mathbb{R}^{d_c}) \rightarrow \mathcal{F}_v(\Omega_v; \mathbb{R}^v), \quad s(x) \mapsto Q(s(x), y)\,, 
\end{equation}
where $Q$ is a finite dimensional neural network from $\mathbb{R}^{d_c} \times \Omega_v \rightarrow \mathbb{R}^v$. 
\end{definition}

Then, when considering the structure in Definition \ref{definition:neural-operator} as an approximator of operators defined in Definition \ref{def:delay_dependent_predictor_operators}, we obtain the following Theorem as a consequence of the continuity proved in Lemma \ref{lemma:predictor-continuity}. 
\begin{theorem}\label{thm:neural-op-approximate-predictor-theorem}
Let assumption \ref{assumption:lipschitz} hold and let $X \in \mathcal{X}$, $U_i \in C^1([0,1];\mathcal{U})$, $i \in [m]$, and $\varphi \in \Phi$, where $\mathcal{X} \subset \mathbb{R}^n$, $\mathcal{U} \subset \mathbb{R}$, and $\Phi \subset \mathbb{R}_+$ are bounded domains. For each $i \in [m]$, fix a compact set $K_i \subset \mathcal{X} \times C^1([0,1];\mathcal{U})^{m-i+1} \times \Phi$. Then, for any $\epsilon > 0$, there exist neural operator approximations $\hat{\mathcal{P}}_i: K_i \to C^1([0,1];\mathbb{R}^n)$ of the predictor operators $\mathcal{P}_i$ from Definition \ref{def:delay_dependent_predictor_operators} such that
\begin{align}
\sup_{(X,U_i,\ldots,U_m,\varphi) \in K_i}&   
\|\mathcal{P}_i(X,U_i,\ldots,U_m,\varphi)(s) \nonumber \\ & - \hat{\mathcal{P}}_i(X,U_i,\ldots,U_m,\varphi)(s)\|_{L^\infty} < \epsilon,
\end{align}
for all $X \in \mathcal{X}$, $U_i \in C^1([0,1);\mathcal{U})$, $i \in [m]$, and $\varphi \in \Phi$.
\end{theorem}

\section{Stability under uniform approximate predictors} \label{sec:stability}
We are now ready to begin studying the stability analysis of the resulting multi-input multi-delay system \eqref{eq:system} under uniform approximate predictors. To do so, we introduce the following three assumptions which are needed for the analysis of the \emph{exact} predictor feedback (\cite{Krstic:2010, Nonlinear_Control_Under_Nonconstant_Delay, Bekiaris-Liberis_Krstic:2016}) and henceforth are reasonable for our study under approximate predictors:
\begin{assumption}\label{assumption:strongly-forward-complete}
\sloppy The system $\dot{X} = f(X, \omega_1, \ldots, \omega_m)$ is strongly forward complete with respect to $\omega = (\omega_1, \ldots, \omega_m)^T$.
\end{assumption}
\begin{assumption}\label{assumption:iss}
\sloppy The system $\dot{X} = f(X, \omega_1 + \kappa_1(X), \ldots, \omega_m + \kappa_m(X))$ is input-to-state stable with respect to $\omega = (\omega_1, \ldots, \omega_m)^T$.
\end{assumption}
\begin{assumption}\label{assumption:strongly-forward-complete-2}
\sloppy The systems $\dot{X} = g_j(X, \omega_{j+1}, \ldots, \omega_m)$, for all $j=1,\dots,m$, with $g_j(X, \omega_{j+1}, \ldots, \omega_m)$ = $f(X, \kappa_1(X), \ldots, \kappa_j(X), \omega_{j+1}, \ldots, \omega_m)$, are strongly forward complete with respect to $\omega = (\omega_{j+1}, \ldots, \omega_m)^T$.
\end{assumption}

The justification for each assumption is in \cite{Bekiaris-Liberis_Krstic:2016}. We are now ready to present the main stability result. 

\begin{theorem}\label{thm:main_result}
Let system \eqref{eq:system} satisfy Assumptions \ref{assumption:lipschitz},\ref{assumption:strongly-forward-complete}-\ref{assumption:strongly-forward-complete-2}, and assume that $\kappa_i$, $i \in [m]$, are stabilizing as in the Assumption \ref{assumption:delay-free-control-laws}. Then, for $\overline{\mathcal{B}} := \min \{ \overline{\mathcal{X}}, \overline{\mathcal{U}} \}$, there exist functions $\alpha_1 \in \mathcal{K}$ and $\beta \in \mathcal{KL}$ such that if $\epsilon < \epsilon^*$ where
\begin{align}
    \epsilon^{*}(\overline{\mathcal{B}}) := \alpha_1^{-1}(\overline{\mathcal{B}})/m, \label{eq:epsilon_star}
\end{align}
and the initial state is constrained to
\begin{align}
    \Gamma(0) \leq \Omega, \label{eq:init_state_restriction}
\end{align}
where
\begin{align}
    \Omega(\epsilon, \overline{\mathcal{B}}) := \alpha_{2}^{-1} \bigg( \overline{\mathcal{B}} - \alpha_{1}(m\epsilon)\bigg), \quad \alpha_{2}(r) = \beta(r, 0), \ r \geq 0, \label{eq:Omega}
\end{align}
then, the closed-loop solutions, under the controller with the neural operator-approximated predictor $U_i(t) = \kappa_i \bigg( \hat{P}_i(t) \bigg)$, $i \in [m]$, satisfy the $\overline{\mathcal{B}}$-semiglobal and $\epsilon$-practical stability estimate
\begin{align}
    \Gamma(t) \leq& \beta(\Gamma(0), t) + \alpha_1(m\epsilon), \quad \forall t \geq 0, \label{eq:main_inequality}
\end{align}
where
\begin{align}
    \Gamma(t) =& |X(t)| + \sum_{i=1}^m \|u_i(t)\|_{L^\infty}, \label{eq:Xi}
\end{align}
with the neural operator semiglobally trained with $\overline{\mathcal{B}} > \alpha_2(\Omega)$ and $\epsilon \in (0, \alpha_1^{-1}(\overline{\mathcal{B}})/m)$.
\end{theorem}

Before proving the result, we briefly discuss its implications. The result states that for all $\epsilon < \epsilon^*$, there exists a region of attraction $\Omega$ whose size grows inversely with $\epsilon$, such that all initial states $\Gamma(0) \leq \Omega$ are $\epsilon$-practically stable. Conversely, for $\epsilon > \epsilon^*$, no such region exists, and the system becomes unstable for any neural network predictor with large errors. Thus, the result has two key perspectives: for neural network engineers, achieving a sufficiently small $\epsilon$ is essential to stabilize the system; for control verification engineers, accurately estimating $\epsilon$ is crucial, as it directly determines both the region of attraction and the size of the invariant stable set. Such estimation is beyond the scope of this study but may be performed using, for example, the $L_2$ testing error.

\begin{proof}
    The proof of Theorem \ref{thm:main_result} will be split over a series of technical Lemmas across two subsections. We begin by introducing an equivalent representation of \eqref{eq:system} representing $m$ delays as individual transport PDEs.
\subsection{Transport PDE equivalent delay system}
Consider the setting in which $U_i$ is invoked with the approximated predictions as in Theorem \ref{thm:neural-op-approximate-predictor-theorem}. Thus, the approximate predictors are given by the operator approximations as 
\begin{subequations}
\begin{align}
    \hat{P}_1(t) =& \hat{\mathcal{P}}_1(X(t), T_{D_1,0}(t)U_1, ..., T_{D_m,D_{m-1}}(t)U_m, D_1)(1), \\
    \hat{P}_2(t) =& \hat{\mathcal{P}}_2(\hat{P}_1(t), T_{D_{2,1},0}(t)U_2, \nonumber \\ &..., T_{D_{m,1},D_{m,2}}(t)U_m, D_{2,1})(1), \\
    &\ \hspace{3cm} \vdots \nonumber \\
    \hat{P}_m(t) =& \hat{\mathcal{P}}_m(\hat{P}_{m-1}(t), T_{D_{m,m-1},0}(t)U_m, D_{m,m-1})(1).
\end{align}
\end{subequations}
Then,  system \eqref{eq:system} can be written equivalently in its ODE-PDE form, originally developed in \cite{Krstic:2010}, as
\begin{subequations}\label{eqs:plant_ode_pde}
    \begin{align}
        \dot{X}(t) =& f(X(t), u_1(0,t), \ldots, u_m(0,t)), \label{eq:state} \\
        \frac{\partial}{\partial t} u_i(x,t) =& \frac{\partial}{\partial x} u_i(x,t), \quad (x,t) \in [0,D_i) \times \mathbb{R}_+,\label{eq:delay_pde} \\
        u_i(D_i,t) =&\kappa_i(\hat{P}_i(t)), \quad i \in [m]. \label{eq:pde_boundary}
    \end{align}
\end{subequations}
The rewritten system is an ODE coupled with $m$ different transport PDEs, where the analytical solution of each transport PDE is given as
\begin{align}
    u_i(x,t) = U_i(t+x-D_i), \quad x \in [0, D_i),
\end{align}
$i \in [m]$. Notice that the transport PDEs \eqref{eq:delay_pde} are defined on growing spatial domains as the number of inputs increases. Hence, the time for the input signal to reach from the input boundary at position $D_i$ to the opposite side at $0$ is exactly $D_i$ seconds with the unit transport speed.

The corresponding predictors then become, in PDE notation, 
\begin{subequations}\label{eqs:ps}
    \begin{align}
        p_1(x,t) =& X(t) + \int_0^x f \bigg( p_1(y,t), u_1(y,t), \ldots, u_m(y,t) \bigg) dy, \label{eq:p1(x,t)}  \nonumber \\
            &\quad x \in [0, D_1) \\
        p_2(x,t) =& p_1(D_1,t) + \int_{D_1}^x f \bigg( p_2(y,t), \kappa_1(p_2(y,t)), u_2(y,t),\nonumber \\ & \ldots, u_m(y,t) \bigg) dy, \label{eq:p2(x,t)} \quad x \in [D_1, D_2) \\
        &\hspace{2cm} \vdots \nonumber \\
        p_m(x,t) =& p_{m-1}(D_{m-1},t) + \int_{D_{m-1}}^x f \bigg( p_m(y,t), \nonumber \\ &  \kappa_1(p_m(y,t)), \ldots, \kappa_{m-1}(p_m(y,t)), u_m(y,t) \bigg) dy, \label{eq:pm(x,t)} \nonumber \\
            &\quad x \in [D_{m-1}, D_m) .
    \end{align}
\end{subequations}
Notice that, the predictors $p_1(D_1, t), ..., p_m(D_m, t)$ are equivalent to the original predictors $P_1(t), ..., P_m(t)$ in \eqref{eqs:P_system}. The key advantage of rewriting \eqref{eq:system} in the form of \eqref{eq:state}, \eqref{eq:delay_pde}, \eqref{eq:pde_boundary} is that analyzing the ODE-PDE cascade becomes significantly simpler then the original ODE under delayed inputs.

\subsection{Stability of the ODE-PDE system}
To analyze the system, we begin by introducing the following backstepping transformation as in \cite{Bekiaris-Liberis_Krstic:2016} to transform the ODE-PDE cascade into what we call the ``target system". The key advantage of the target system is that it is feasible to show that the system is ISS-stable \cite{iasson-iss}. This stability result for the target system can then be transformed into a stability result for the original system via an inverse backstepping transformation, completing the main result. 
\begin{lemma}\label{lemma:backstepping-transform-to-target-system}
The backstepping transformations for the transport PDEs $u_i$ defined by
\begin{subequations}
\begin{align}
    w_1(x,t) =& u_1(x,t) - \kappa_1\bigg(p_1(x,t)\bigg), \quad x \in [0, D_1) \label{eq:w1} \\
    w_2(x,t) =& u_2(x,t) - 
        \begin{cases} 
            \kappa_2\bigg(p_1(x,t)\bigg), \quad x \in [0, D_1) \\
            \kappa_2\bigg(p_2(x,t)\bigg), \quad x \in [D_1, D_2)
        \end{cases} \label{eq:w2} \\
    &\ \vdots \nonumber \\
    w_m(x,t) =& u_m(x,t) - 
        \begin{cases} 
            \kappa_m\bigg(p_1(x,t)\bigg), \quad x \in [0, D_1) \\
            \kappa_m\bigg(p_2(x,t)\bigg), \quad x \in [D_1, D_2) \\
            \ \vdots \\
            \kappa_m\bigg(p_m(x,t)\bigg), \quad x \in [D_{m-1}, D_m)
        \end{cases}, \label{eq:wm}
\end{align}
\end{subequations}
transform the plant \eqref{eqs:plant_ode_pde} to the target system $(i \in [m])$
\begin{subequations}\label{eqs:target}
    \begin{align}
        \dot{X}(t) &= f \bigg( X(t), \kappa_1(X(t)) + w_1(0,t),\nonumber \\ & \dots, \kappa_m(X(t)) + w_m(0,t) \bigg),  \label{eq:target_ode} \\
        \partial_t w_i(x,t) &= \partial_x w_i(x,t), \quad (x,t) \in (0,D_i) \times [0,\infty) \label{eq:target_pde} \\
        w_i(D_i,t) &= \kappa_i(\hat{P}_i(t)) - \kappa_i(P_i(t)), \quad t \in [0,\infty). \label{eq:target_pde_boundary}
    \end{align}
\end{subequations}
\end{lemma}

\begin{proof}
    The proof follows exactly that of \cite[Lemma 1]{Bekiaris-Liberis_Krstic:2016} except that the boundary condition on \eqref{eq:pde_boundary} is with the approximate predictor and henceforth the resulting boundary condition of the target system is obtained via direct substitution, yielding \eqref{eq:target_pde_boundary}. 
\end{proof}

Additionally, there exists an invertible set of transformations from the $w$ to the $u$ system characterized by

\begin{lemma}\label{lemma:inverse-backstepping-transforms}
    For all $i \in [m]$, the inverse backstepping transformation for the PDEs $w_i$ defined by 
    \begin{subequations}
    \begin{align}
        u_1(x,t) =& w_1(x,t) + \kappa_1\bigg(\pi_1(x,t)\bigg), \quad x \in [0, D_1) \label{eq:inv_w1} \\
        u_2(x,t) =& w_2(x,t) +
            \begin{cases}
                \kappa_2\bigg(\pi_1(x,t)\bigg), \quad x \in [0, D_1) \\
                \kappa_2\bigg(\pi_2(x,t)\bigg), \quad x \in [D_1, D_2)
            \end{cases} \label{eq:inv_w2} \\
        &\ \vdots \nonumber \\
        u_m(x,t) =& w_m(x,t) -
            \begin{cases}
                \kappa_m\bigg(\pi_1(x,t)\bigg), \quad x \in [0, D_1) \\
                \kappa_m\bigg(\pi_2(x,t)\bigg), \quad x \in [D_1, D_2) \\
                \ \vdots \\
                \kappa_m\bigg(\pi_m(x,t)\bigg),\\ \qquad x \in [D_{m-1}, D_m)
            \end{cases} \label{eq:inv_wm}
    \end{align}
    \end{subequations}
    where
    \begin{subequations}
    \begin{align}
        \pi_1(x,t) =& X(t) + \int_0^x f \bigg (\pi_1(y,t), w_1(y,t) + \kappa_1 \bigg( \pi_1(y,t) \bigg), \nonumber \\ & \ldots, w_m(y,t) + \kappa_m \bigg( \pi_1(y,t) \bigg) \bigg) dy, \label{eq:pi1} \\
            &\qquad x \in [0, D_1] \nonumber \\
        \pi_i(x,t) =& \pi_{i-1}(D_{i-1},t) + \int_{D_{i-1}}^x f \bigg( \pi_i(y,t), \nonumber \\ & w_1(y,t) + \kappa_1 \bigg( \pi_i(y,t) \bigg), \ldots, \nonumber \\ &  w_m(y,t) + \kappa_m \bigg( \pi_m(y,t) \bigg) \bigg) dy, \label{eq:pii} \\
            &\qquad x \in [D_{i-1}, D_i], \quad i=2,\dots,m. \nonumber
    \end{align}
    \end{subequations}
    transforms the plant \eqref{eq:target_ode}, \eqref{eq:target_pde}, \eqref{eq:target_pde_boundary} into the plant \eqref{eq:state}, \eqref{eq:delay_pde}, \eqref{eq:pde_boundary}. 
\end{lemma}

\begin{proof}
    As in Lemma \ref{lemma:backstepping-transform-to-target-system}, the proof follows \cite[Lemma 2]{Bekiaris-Liberis_Krstic:2016} with the appropriate substitutions of the boundary condition due to the approximate predictors. 
\end{proof}

We are now ready to provide an ISS stability result for the target ODE-PDE system.

\begin{lemma}\label{lem:lemma_3_of_6}
Let the system \eqref{eqs:target} satisfy the Assumption \ref{assumption:iss}. Then, there exists a class $\mathcal{KL}$ function $\beta_1$ and a class $\mathcal{K}$ function $\gamma_1$ such that the following holds for all $t \geq 0$
\begin{align}
    \breve{\Gamma}(t) \leq& \nonumber  \beta_1 \bigg( \breve{\Gamma}(0), t \bigg) \\ &+ \gamma_1 \bigg( \sup_{0 \leq \tau \leq t} \bigg\{ \sum_{i=1}^m |\kappa_i(\hat{P}_i(t)) - \kappa_i(P_i(t))| \bigg\} \bigg), \label{eq:lemma3_Gamma_bar}
\end{align}
where
\begin{align}
    \breve{\Gamma}(t) :=& |X(t)| + \sum_{i=1}^m \|w_i(t)\|_{L^\infty}. \label{eq:Gamma_bar_def}
\end{align}
\end{lemma}
\begin{proof}
From Assumption \ref{assumption:iss} and the definition of input-to-state stability in \cite{SontagISS:1995}, it follows that there exist function $\beta_2 \in \mathcal{KL}$ and $\gamma_2 \in \mathcal{K}$ such that
\begin{align}
    |X(t)| \leq& \beta_2(|X(0)|, t) + \gamma_2 \bigg( \sup_{0 \leq \tau \leq t} \bigg\{ \sum_{i=1}^m |w_i(0,\tau)| \bigg\} \bigg) \label{eq:w_iss}\,, 
\end{align}
for all $t \geq 0$. Using the ISS estimate of the transport PDE from \cite[Lemma 6]{Luke_Peijia:2025}, for any constant $c> 0$, for all $i \in [m]$, and for all times $t > 0$, we have that 
\begin{align}
    \|w_i(t)\|_{L^\infty} \leq& \exp(c(D_i-t)) \|w_i(0)\|_{L^\infty} \nonumber  \\ &+ \exp(cD_i) \sup_{0 \leq \tau \leq t} |w_i(D_i,\tau)|. \label{eq:wi_pde_iss}
\end{align}
From \eqref{eq:w_iss}, we have
\begin{align}
    |X(t)| \leq& \beta_2(|X(0)|, t) + \gamma_2 \bigg( \sup_{0 \leq \tau \leq t} \bigg\{ \sum_{i=1}^m |w_i(0,\tau)| \bigg\} \bigg) \nonumber \\
        \leq& \beta_2(|X(0)|, t) + \gamma_2 \bigg( \sup_{0 \leq \tau \leq t} \bigg\{ \sum_{i=1}^m \|w_i(\tau)\|_{L^\infty} \nonumber \\ & + \sum_{i=1}^m |w_i(D_i,\tau)| \bigg\} \bigg) \label{eq:lemma3_x_iss},
\end{align}
and, from \eqref{eq:wi_pde_iss} applied to all $i \in [m]$, we have
\begin{align}
    \sum_{i=1}^m \|w_i(t)\|_{L^\infty}  
        \leq& \beta_3 \bigg( \sum_{i=1}^m \|w_i(0)\|_{L^\infty}, t \bigg) \nonumber \\ & + \gamma_3 \bigg( \sup_{0 \leq \tau \leq t} \bigg\{ \sum_{i=1}^m |w_i(D_i,\tau)| \bigg\} \bigg), \label{eq:lemma3_w_iss}
\end{align}
where $\beta_3 \in \mathcal{KL}$, and $\gamma_3 \in \mathcal{K}$. Now, treating $X(t)$ and $\sum_{i=1}^m \|w_i(t)\|_{L^\infty}$ as two interconnected subsystems with $\sup_{0 \leq \tau \leq t} \sum_{i=1}^m|w_i(D_i,\tau)|$ as their input, we can apply the classical results of cascaded ISS systems (See \cite[Lemme C.4]{KKK} or \cite{402246}) using the forms \eqref{eq:lemma3_x_iss}, \eqref{eq:lemma3_w_iss} to obtain 
\begin{align}
    |X(t)|\nonumber  + \sum_{i=1}^m \|w_i(0)\|_{L^\infty} \leq& \beta_1 \bigg( |X(0)| + \sum_{i=1}^m \|w_i(0)\|_{L^\infty}, t \bigg) \nonumber \\&\!\!\!+ \gamma_1 \bigg( \sup_{0 \leq \tau \leq t} \bigg\{ \sum_{i=1}^m |w_i(D_i,\tau)| \bigg\} \bigg), \label{eq:x_and_w_after_c4}
\end{align}
where $\beta_1 \in \mathcal{KL}$, and $\gamma_1 \in \mathcal{K}$. Lastly, by defining the function $\breve{\Gamma}(t) = |X(t)| + \sum_{i=1}^m \|w_i(t)\|_{L^\infty}$ and using \eqref{eq:target_pde_boundary} completes the proof.
\end{proof}

To complete the proof of Theorem \ref{thm:main_result}, the exact same approach to convert from the bound on the target system to the original system using \cite[Lemmas 4, 5, 6]{Bekiaris-Liberis_Krstic:2016} can be applied to \eqref{eq:lemma3_Gamma_bar} using properties of class $\mathcal{K}$ functions. Thus, there exists $\beta_5 \in \mathcal{KL}$, $\gamma_5 \in \mathcal{KL}$ such that
\begin{align}
    \Gamma(t) \leq& \beta_5 \bigg( \Gamma(0), t \bigg) \nonumber  \\ & + \gamma_5 \bigg( \sup_{0 \leq \tau \leq t} \bigg\{ \sum_{i=1}^m |\kappa_i(\hat{P}_i(t)) - \kappa_i(P_i(t))| \bigg\} \bigg), \label{eq:original-system-bound}
\end{align}
Using the Lipschitz constant $C_\kappa$ as in \eqref{eq:c_kappa} with the universal approximation theorem (Theorem \ref{thm:neural-op-approximate-predictor-theorem}), and property of class $\mathcal{K}$ functions, one achieves Theorem \ref{thm:main_result}. 
\end{proof}

\section{Numerical simulations} \label{sec:numerical}
To validate our theoretical results, we develop neural operator approximate predictors for the unicycle system with distinct steering and drive delays. Consider the following: 
\begin{subequations}
\begin{align}
    x_(t) &= U_2(t-D_2)\cos(\theta(t)) \label{eq:unicycle_x1_dot} \\
    y(t) &= U_2(t-D_2)\sin(\theta(t)) \label{eq:unicycle_x2_dot} \\
    \theta(t) &= U_1(t-D_1), \label{eq:unicycle_x3_dot}
\end{align}
\end{subequations}
where $x, y$ are the 2D coordinate positions of the unicycle, $\theta$ is spatial orientation, $U_1$ is the turning rate, $U_2$ is linear speed, and $D_2$ is the longer breaking delay and $D_1$ is the short steering delay. We consider the time-varying controller of \cite{POMET1992147} given as presented in \cite[Section V]{Bekiaris-Liberis_Krstic:2016}.

In our numerical experiment, we consider delays of $D_1 = 0.25$s and $D_2= 0.6$s with a time step $dt=0.001$s from $0$ to $T=10$s of simulation time. 
As both a baseline and to create a dataset for training the neural operator, we simulate $100$ trajectories using the numerical approach of \cite{Predictor_Feedback_For_Delay_Systems}, applying additive uniform noise $\mathcal{U}(-0.2, 0.2)$ to the state at each timestep. From each trajectory, we collect all predictor input/output pairs, resulting in a total dataset of 200,000 pairs. All trajectories start from the same initial conditions $(x, y, \theta) = \bm{0}$, with zero initial input $U_j(s) = 0$, $s \in [-D_j, 0]$, $j \in [2]$. The dataset generation process takes about an hour on a standard laptop. For training, we consider two popular neural operator architectures—FNO and DeepONet—which are trained in a supervised manner(10 minutes, Nvidia 4060 GPU). \textcolor{blue}{See \cite{ProjectGitHub} for all code, datasets, models, and hyperparameters.}

\begin{figure*}[t]
    \centering
    \includegraphics[trim=0 5px 0 7px, clip]{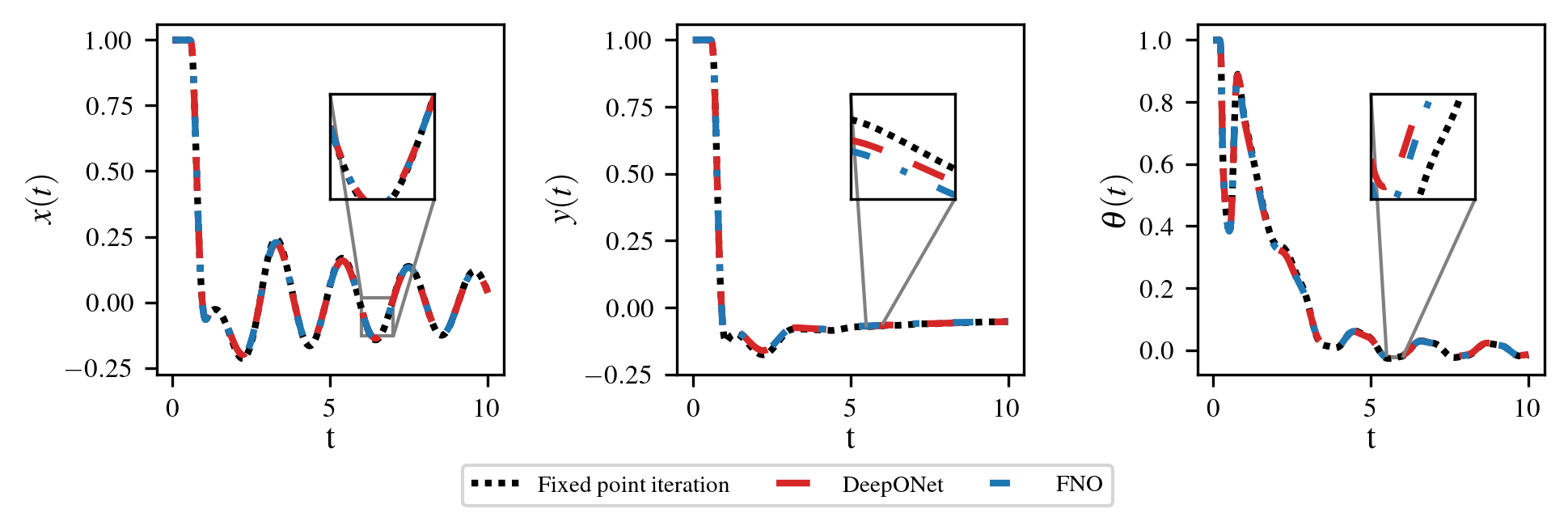}
    \caption{Example evaluation trajectories with fixed point, DeepONet, and FNO approximated predictors.}
    \label{fig:single_trajectory}
    \vspace{-1em}
\end{figure*}

In Table \ref{tab:performance_errors}, we present the performance of the neural approximated predictors.
All approaches achieve strong performance with errors on the magnitude of $\sim\!10^{-6}$. 
Furthermore, to accurately evaluate the neural approximated predictor operator in the feedback loop, we simulate 25 new trajectories for each approximation approach and present the average trajectory residuals with respect to the stationary equilibrium point, $(x,y,\theta)=\bm{0}$ in Table \ref{tab:performance_errors}. Furthermore, to showcase the similarity between both approaches, we provide an example of one of the $25$ trajectories in Figure \ref{fig:single_trajectory}.

Table \ref{tab:performance_times} highlights the computational advantages of neural operator approximation of predictors over fixed point iteration, particularly at finer spatial resolutions. While fixed point iteration’s computation time increases significantly as resolution improves, FNO maintains nearly constant runtimes, and DeepONet shows moderate variation. These traits enable both neural operators to outperform traditional methods as resolution becomes finer. Notably, their speedups are largely due to deployment on specialized GPU hardware compared to CPU-based fixed point iteration.

\begin{table}[H]
\centering
\caption{
  The performance of predictor approximation approaches across training and validation datasets. The average trajectory tracking error is the summed $L^2$ error over time averaged over the 25 different test scenarios (with randomized initial conditions).}
\label{tab:performance_errors}
\resizebox{0.5\textwidth}{!}{
\begin{tabular}{@{}*{5}{c}@{}}
    \toprule
    \multirow{2}{*}{\makecell[c]{Approximation\\methods}} & 
    \multirow{2}{*}{\makecell[c]{Parameters}} & 
    \multirow{2}{*}{\makecell[c]{Training $L_2$\\error $(\times10^{-6})\downarrow$}} & 
    \multirow{2}{*}{\makecell[c]{Validation $L_2$\\error $(\times10^{-6})\downarrow$}} & 
    \multirow{2}{*}{\makecell[c]{Avg. trajectory\\$L_2$ error $\downarrow$}} \\
    & & & & \\
    \midrule
    \makecell[c]{Fixed point\\approximation} & \makecell[c]{---} & \makecell[c]{---} & \makecell[c]{---} & 5.131 \\
    FNO & 42211 & 0.33 & 0.25 & 5.144 \\
    DeepONet & 1124824 & 3.00 & 2.46 & 4.874 \\
    \bottomrule
\end{tabular}
}
\end{table}

\vspace{-1em}

\begin{table}[H]
\centering
\caption{
  Computation time (ms) averaged over $1,000$ predictions. 
}
\label{tab:performance_times}
\resizebox{0.5\textwidth}{!}{
\begin{tabular}{@{}*{6}{c}@{}}
    \toprule
    \multirow{2}{*}{\makecell[c]{Approximation\\methods}} & 
    \multicolumn{4}{c}{\makecell[c]{Computation time (ms)}} & 
    \multirow{2}{*}{\makecell[c]{Speedup $\uparrow$}} \\
    \cmidrule(lr){2-5}
    & \makecell[c]{dx = 0.01} & \makecell[c]{dx = 0.005} & \makecell[c]{dx = 0.001} & \makecell[c]{dx = 0.0005} & \\
    \midrule
    \makecell[c]{Fixed point\\approximation} & 8.762 & 17.735 & 97.676 & 218.740 & \makecell[c]{---} \\
    FNO & 30.850 & 29.664 & 34.653 & 31.093 & 7.04$\times$ \\
    DeepONet & 7.245 & 9.704 & 10.816 & 23.126 & 9.46$\times$ \\
    \bottomrule
\end{tabular}
}
\end{table}

\section{Conclusion and Future Work}
In conclusion, we extended the neural operator predictor framework to multi-input nonlinear systems with distinct delays, proving existence and semiglobal practical stability (Theorems \ref{thm:neural-op-approximate-predictor-theorem} and \ref{thm:main_result}). Our results, validated by unicycle simulations and consistent with single-delay cases \cite{Luke_Peijia:2025}, highlight the balance between efficiency and accuracy and applies to any approximate predictor satisfying uniform error bounds. This opens the door for more complex time-varying approximate predictor methodologies such as those with noise and output feedback designs. 
\bibliography{references}
\bibliographystyle{abbrv}

\end{document}